\title{Permutation testing in high-dimensional linear models: an empirical investigation\footnote{Accepted for publication in \emph{Journal of Statistical Computation and Simulation}}}    
\author{Jesse Hemerik\footnote{Biometris, Wageningen University \& Research, The Netherlands}
\footnote{Address for correspondence: Jesse Hemerik, Biometris, Wageningen University \& Research, P.O. Box 16, 6700 AA Wageningen, The Netherlands. e-mail: jesse.hemerik@wur.nl},
  \phantom{.}Magne Thoresen\footnote{Oslo Centre for Biostatistics and Epidemiology, University of Oslo, Norway}
 \phantom{.}and
Livio Finos\footnote{Department of Developmental Psychology and Socialization, University of Padua, Italy}
}

\documentclass[11pt]{article}

\usepackage{amsmath}
\usepackage{amsfonts}
\usepackage{bm}
\usepackage{bbm}
\usepackage{amsthm}
\usepackage{comment}
\usepackage[nottoc]{tocbibind}
\usepackage{graphicx}

\usepackage{subcaption}

\usepackage{changepage}
\usepackage{natbib}
\usepackage{abstract}
\usepackage[colorlinks=true,citecolor=blue,runcolor=blue,linkcolor=blue, pdfborder={0 0 0}]{hyperref}

\usepackage{tikz}

\usepackage{algorithm}
\usepackage{algorithmic}

\definecolor{darkblue}{rgb}{0.0, 0.0, 0.55}

\theoremstyle{plain}

\newtheorem{proposition}{Proposition}

\theoremstyle{definition}

\usepackage[margin=3cm]{geometry}

\newcommand{\re}{\bm{R}}
\newcommand{\rer}{\tilde{\bm{R}}}
\newcommand{\pro}{\bm{H}}
\newcommand{\pror}{\tilde{\bm{H}}} 
\newcommand{\Y}{\bm{Y}}
\newcommand{\X}{\bm{X}}
\newcommand{\Z}{\bm{Z}}
\newcommand{\Pe}{\bm{P}}

\newcommand{\A}{\bm{A}}
\newcommand{\B}{\bm{B}}
\newcommand{\U}{\bm{U}}
\newcommand{\D}{\bm{D}}
\newcommand{\V}{\bm{V}}

\begin{document}
\maketitle

\begin{abstract}
\noindent Permutation testing in linear models, where the number of nuisance coefficients is smaller than the sample size, is a well-studied topic.
The common approach of such tests is to permute residuals after regressing on the nuisance covariates.
Permutation-based tests are valuable in particular  because they  can be  highly robust to violations of the standard linear model, such as non-normality and heteroscedasticity.
Moreover, in some cases they  can be combined with existing, powerful permutation-based multiple testing methods. 
Here, we propose permutation tests for models where the number of nuisance coefficients exceeds the sample size.
The performance of the novel tests is investigated with simulations.
In a wide range of simulation scenarios our proposed permutation methods provided appropriate type I error rate control, unlike some competing  tests, while having good power.
\\
\\
\emph{keywords:} Permutation test; Group invariance test; High-dimensional inference; Heteroscedasticity; Semi-parametric
\end{abstract}

\section{Introduction}

We consider the problem of testing hypotheses about coefficients  in   linear models, where the outcome may be non-Gaussian and heteroscedastic, and   the number of nuisance coefficients exceeds the sample size. 
By the nuisance coefficients we mean the coefficients that are not tested by the particular test at hand, but still need to be dealt with  since they lead to confounding effects.
In recent decades,  the literature  on permutation methods has strongly expanded \citep{tusher2001significance,  meinshausen2011asymptotic,  hemerik2018false, ganong2018permutation, berrett2018conditional,he2019permutation,albajes2019voxel, hemerik2019permutation, rao2019permutation}.
While the permutation test dates far back \citep{fisher1936coefficient}, most of the permutation tests in the presence of nuisance were published in the last four decades. To our knowledge, the existing methods are limited to low-dimensional nuisance. 
For the high-dimensional case, an approach similar to a permutation test  is proposed in \citet{dezeure2017high}.

Permutation tests for low-dimensional linear models are valuable for two main reasons. First, they are robust to  violations of certain standard assumptions, such as  normality and homoscedasticity \citep{winkler2014permutation,hemerik2020robust}.
Second, when the outcome is multidimensional, a permutation-based test can be combined with existing permutation-based multiple testing methods, which tend to be relatively powerful, since they take into account the dependence structure of the outcomes \citep{meinshausen2006false, meinshausen2011asymptotic, hemerik2018false,hemerik2019permutation}.
For example, under strong positive dependence among \emph{p}-values, the Bonferroni-Holm multiple testing method \citep{holm1979simple}  is greatly improved by a permutation method \citep{
westfall1993resampling}.

For the   low-dimensional general linear model, with identity link but not necessarily Gaussian or homoscedastic residuals, several different permutation tests have been proposed. 
The main approach that these methods  have in common,   is to permute residuals after regressing on the nuisance covariates.
 For overviews of the available methods, see \citet{anderson1999empirical}, \citet{anderson2001permutation}, \citet{winkler2016faster} and in particular \citet{winkler2014permutation}. 
Among the existing permutation methods, the Freedman-Lane approach \citep{freedman1983nonstochastic} is most commonly used and provides excellent power and type I error control.

Because  the existing permutation tests require estimating the nuisance coefficients using maximum likelihood, these methods cannot be used when the number of covariates exceeds the sample size.  
In recent years, important theoretical results have been published on testing in such 
high-dimensional linear models.  Several of these tests have proven asymptotic properties. In particular, the method in \citet{zhang2014confidence} has been shown to be asymptotically optimal under certain assumptions \citep{van2014asymptotically}.  \citet{dezeure2017high} propose a bootstrap approach, which is related to the method in  \citet{zhang2014confidence}. 
 Software implementations of tests for high-dimensional models include those described in \citet{dezeure2015high} and \citet{Chernozhukov2016}.

Testing in high-dimensional linear models is very challenging, because a large number of unknown nuisance effects needs to be dealt with, using a relatively small sample size. Consequently,  tests tend to sacrifice much power compared to the situation where all nuisance coefficients would be known. 
Further, the  asymptotic properties of the mentioned methods rely on complex assumptions and sparsity.
The test by \citet{zhang2014confidence} can  be rather anti-conservative in settings where a substantial fraction of the coefficients are non-zero. Moreover, these methods are not based on permutations. Hence they do not generally have the above-mentioned advantages, such as robustness against certain violations of the standard linear model. An exception is the  bootstrap method in \citet{dezeure2017high}, which tends to be more robust to such violations.

We propose two novel tests, which, to our knowledge, are the first permutation tests in the presence of high-dimensional nuisance.
One is an extension of the low-dimensional method in \citet{freedman1983nonstochastic} and the other is somewhat related to a method by  Kennedy \citep{kennedy1995randomization, kennedy1996randomization}. Further, we allow the tested parameter to be multi-dimensional, unlike many existing methods. 
Using simulations we show  that our methods  provide appropriate type I error rate control in a wide range of situations.   In particular, we  illustrate empirically that our tests have the above-mentioned robustness  properties. 
The methods in this paper have been implemented in the R package \emph{phd}, available on CRAN.

This paper is built up as follows. 
In Section \ref{secldn} we discuss  permutation testing in settings with low-dimensional nuisance.  This section contains some novel observations that will be used in Section \ref{sechd}.
There, we propose  permutation tests for high-dimensional settings.  We assess the performance of our methods with simulations in Section \ref{secsims}. An analysis of real data is in Section \ref{secdata}.

\section{Low-dimensional nuisance} \label{secldn}

\subsection{Notation and basic ideas} \label{secnota}
We consider the general linear model
$$\bm{Y}=\bm{X}\bm{\beta}+\bm{Z}\bm{\gamma} + \bm{\epsilon},$$
where $\bm{X}$ is a $n\times d$ matrix of covariates of interest, $\bm{Z}$ an $n\times q$  matrix of nuisance covariates and $\bm{\epsilon}$ an $n$-vector of i.i.d. errors with mean $0$ and non-zero variance, which are independent of  the covariates.   Here the rows of $\bm{X}$, $\bm{Z}$ and  $\bm{Y}$ are i.i.d..      The matrix $\bm{Z}$ is assumed to have full rank with probability $1$.   
The parameter  $\bm{\beta}\in \mathbb{R}^d$ is of interest and $\bm{\gamma} \in \mathbb{R}^q$ is a nuisance parameter.   We want to test the null hypothesis $H_0: \bm{\beta}=\bm{0}\in \mathbb{R}^d$. Here $\bm{0}$ might be replaced by another constant: the extension is straightforward.

Let $w$ be a positive integer, which will denote the number of random  permutations or other transformations. In this paper, all permutation \emph{p}-values  are of the form
\begin{equation} \label{formulap}
p=w^{-1}\big|\{1\leq j \leq w: T_j\geq T_1\}\big|,
\end{equation}
or, in case of a two-sided test where both small and large values of $T_1$ are evidence against $H_0$,
\begin{equation} \label{formulap2}
p=2 w^{-1}  \min\Big\{\big|\{1\leq j \leq w: T_j\geq T_1\}\big|,\big|\{1\leq j \leq w: T_j\leq T_1\}\big|\Big\}.
\end{equation}
 Here $T_1,...,T_w\in \mathbb{R}$ are statistics whose definition depends on the particular permutation method. They are specified in the sections below. For every $2\leq j \leq w$, the statistic $T_j$ corresponds to the $j$-th permutation. 
The statistic $T_1$ is based on the original, unpermuted data.  All existing and novel methods in this paper only differ with respect to how $T_1,...,T_w$ are computed.

Although we will often write  `permutation', sign-flipping of residuals can also be used \citep{winkler2014permutation}. 
The existing methods, as well as the novel methods in this paper, consist of the following steps.
\vskip3mm

 \begin{enumerate}
\item Compute a test statistic $T_1$ based on the original data.
\item Compute a test statistic $T_2$ in a similar way, but after randomly permuting certain residuals. Repeat to obtain $T_3,...,T_w$.
\item The \emph{p}-value equals \eqref{formulap} or  \eqref{formulap2}.
\end{enumerate}
\vskip3mm

Most of the existing permutation methods use residualization of $\bm{Y}$ or $\bm{X}$ with respect to the nuisance $\bm{Z}$. 
In the low-dimensional situation, the residual forming matrix is
$$\re= \bm{I}-\pro= \bm{I}- \Z(\Z'\Z)^{-1}\Z'.$$
When $d=1$ we will sometimes consider  $\re\bm{X}\in \mathbb{R}^n$, which is assumed to be nonzero with probability 1. In Section \ref{secldn} we assume $\Z$ contains a column of $1$'s. This  implies that the entries of $\re\X$ and $\re\Y$ sum up to 0.

Note that if we use permutation, we can write the transformed residuals  as
$\bm{P}\re\bm{Y},$
where $\bm{P}$ is an $n \times n$  matrix with exactly one $1$ in every row and column and elsewhere 0's. In case of sign-flipping, $\bm{P}$ is instead an $n \times n$ diagonal matrix with diagonal elements in $\{1,-1\}$ \citep{winkler2014permutation}.
We write $\Pe_1,...,\Pe_w$ to distinguish the $w$ random permutation matrices. Here $\Pe_1$ is the identity matrix and  $\Pe_2,...,\Pe_w$ are random.

\subsection{Choice of test statistics} \label{fl}

Here we discuss the choice of test statistics within the permutation method of Freedman and Lane  \citep{freedman1983nonstochastic,winkler2014permutation}. 
The purpose of this section is to discuss some existing and novel results that we will use in Section \ref{sechd}.

The Freedman-Lane permutation method  is known to provide excellent type I error control, with both its level and power staying very close to the parametric \emph{F}-test, under the Gaussian model.
The test statistic $T_1$ is  based on the unpermuted model $\Y=\X\bm{\beta}+\Z\bm{\gamma}+\epsilon$.
The other statistics are obtained after randomly transforming the residuals. That is,
for $2\leq j \leq w$ the  statistic $T_j$ is based on the model  $(\Pe_j\re+\pro)\Y=\X\bm{\beta}+\Z\bm{\gamma}+\epsilon$, where the same test  statistic, say $T$, is used as for computing $T_1$.
Thus\begin{equation} \label{T1FL}
T_1=T(\X,\Z,\Y),
\end{equation}
\begin{equation} \label{TjFL}
T_j=T\big(\X,\Z,(\Pe_j\re+\pro)\Y\big),
\end{equation}
 where $T$ is a suitable test statistic, the choice of which we now discuss.

It is usually  important to take $T$ to be an asymptotically pivotal statistic,  i.e., a statistic whose asymptotic null distribution does not depend on any unknowns under $H_0$ (\citeauthor{kennedy1996randomization}, \citeyear{kennedy1996randomization}, p.926-927,  \citeauthor{winkler2014permutation}, \citeyear{winkler2014permutation},  p.382, \citeauthor{hall1989effect}, \citeyear{hall1989effect}, \citeauthor{hall1991two}, \citeyear{hall1991two}). 
A pivotal statistic $T$ will always involve estimation of the nuisance parameters.  Thus, after every permutation, the nuisance parameters  need to be estimated anew. Examples of pivotal test statistics are the  \emph{F}-statistic and Wald statistic. These are equivalent: the resulting  permutation \emph{p}-value \eqref{formulap} is the same.

In case $X$ is one-dimensional, the \emph{F}-statistic is  also equivalent to the square of the \emph{partial correlation} \citep{fisher1924distribution, agresti2015foundations}, which is used in  \citet{anderson2001permutation}. The partial correlation is the sample Pearson correlation of  $\re \Y$ and $\re \X$,
\begin{equation} \label{parcor}
\rho\big( \re \Y,\re\X  \big)=    \frac{  (\re\Y)'\re\X  }{ \sqrt{ \sum_i (\re\Y)_i^2  \sum_i (\re\X)_i^2 } }.
\end{equation}
Here we used that the sample means of $\re\Y$ and $\re\X$ are $0$.
If we use the partial correlation in the Freedman-Lane permutation test, this means that we take 
$T(\X,\Z,\Y)= \rho\big( \re \Y,\re\X  \big),$
so that  \eqref{T1FL} and \eqref{TjFL} become
\begin{equation} \label{T1OLS}
T_1= \rho\big( \re \Y,\re\X  \big)
\end{equation}
\begin{equation}  \label{TjOLS}
 \quad T_j=  \rho\big( \re (\Pe_j\re+\pro)\Y     ,\re\X  \big),
\end{equation}
where $\re (\Pe_j\re+\pro)$ could be simplified to $\re \Pe_j\re$, since $\re\pro=\bm{0}$.

The numerator in \eqref{parcor} is 
$$(\re\Y)'\re\X=\Y'\re'\re\X=\Y'\re'\X=(\re\Y)'\X,$$ 
so that \eqref{parcor} equals
\begin{equation} \label{semipc2}
  \frac{  (\re\Y)'\X  }{  \sqrt{ \sum_i (\re\Y)_i^2  \sum_i (\re\X)_i^2 }}. 
\end{equation}
The Freedman-Lane  test with $T$ defined by \eqref{semipc2} remains unchanged if in \eqref{semipc2} we replace  $ \sum_i (\re\X)_i^2$ by $1$ or by the constant  $\sum_i \X_i^2$. Indeed, $T_1,...,T_w$ will just be multiplied by the same constant.
Thus, with respect to the permutation test, the statistic \eqref{parcor} is equivalent to 
\begin{equation} \label{semiparcor}
 \frac{  (\re\Y)'\X  }{ \sqrt{ \sum_i (\re\Y)_i^2  \sum_i \X_i^2  }}.
\end{equation}
If $\X$ has been centered around $0$, then this equals 
\begin{equation}    \label{T1OLSsemi}
\rho\big( \re \Y,\X  \big) = \frac{  (\re\Y)'(\X-\bm{\mu}_x)  }{ \sqrt{ \sum_i (\re\Y)_i^2  \sum_i (\X_i-\bm{\mu}_x)^2  }},
\end{equation}
 where $\bm{\mu}_x$ denotes the $n$-vector with entries equal to the   sample mean of $\X$.
This is  the sample correlation of $\re\Y$ and $\X$ and  is called the \emph{semi-partial correlation}.
 Thus, if $\X$ is centered, using the partial correlation is equivalent to using the semi-partial correlation. 

If we take $T$ to be the semi-partial correlation, then \eqref{T1FL} and \eqref{TjFL} become
$T_1= \rho\big( \re \Y,\X  \big)$ and 
\begin{equation}  \label{TjOLSsemi}
 \quad T_j=  \rho\big( \re (\Pe_j\re+\pro)\Y     ,\X  \big)=  
\frac{  \big(\re (\Pe_j\re+\pro)\Y \big)'(\X-\bm{\mu}_x)  }{ \sqrt{ \sum_i \big( \re (\Pe_j\re+\pro)\Y\big)_i^2  \sum_i (\X_i-\bm{\mu}_x)^2  }},
\end{equation} where $\re (\Pe_j\re+\pro)$ could be simplified to $\re \Pe_j\re$.
Note that we could simply  leave the constant $\sum_i (\X_i-\bm{\mu}_x)^2$ out without changing the result of the permutation test.
Although for centered $\X$ the  statistics \eqref{parcor} and \eqref{T1OLSsemi}   are  equivalent, their counterparts in the high-dimensional setting are not, as will be discussed in Section \ref{secflhd}.

\section{High-dimensional nuisance} \label{sechd}  

When the nuisance parameter $\bm{\gamma}$ has dimension $q\geq n$, the existing permutation methods cannot be used. Here,  these approaches are adapted to obtain tests which can account for high-dimensional nuisance. We first consider the case that $X$ is one-dimensional, i.e., $d=1$.   The case that $d>1$ is discussed in Section \ref{multidimbeta}. 
We assume that the entries of $\Y$, $\X$ and $\Z$ have expected value $0$. Consequently, the intercept is $0$.

All existing tests rely on residualization steps, where $\Y$ or $\X$  is regressed on $\Z$. 
A natural way to adapt this step to the high-dimensional setting,  is to instead estimate the residuals using some  type of elastic net regularization.   
We will consider ridge regression. For minimizing prediction error, ridge regression is often preferrable to Lasso, principal components regression, variable subset selection and partial least squares \citep{hastie2009elements,frank1993statistical}.

Compared to the existing methods, including the Freedman-Lane approach discussed  in Section \ref{fl}, using ridge regression comes down to replacing the  projections $\hat{\Y}=\pro \Y$ and $\hat{\X}=\pro \X$ by ridge estimates
$ \pror_{\lambda}\Y$
and   $   \pror_{\lambda_X}\X   $, with $\lambda, \lambda_X>0$. 
Here, for $\lambda'>0$,
\begin{equation} \label{prorl}
  \pror_{\lambda'}=\Z(\Z'\Z+\lambda' \bm{I}_q)^{-1}\Z',
\end{equation}
which satisfies 
$$\pror_{\lambda'}\Y= \Z  \text{argmin}_{\bm{\gamma}}\Big(\Vert  \Y-\Z\bm{\gamma}  \Vert_2^2 +\lambda' \Vert \bm{\gamma} \Vert_2^2     \Big)$$ and similarly for $\X$.
The values $\lambda, \lambda_X$ are the regularization parameters, whose selection  will be discussed. 
Using ridge regression, the residuals become  $\rer_{\lambda} \Y$ and $\rer_{\lambda_X} \X$, where $\rer_{\lambda}=(\bm{I}-\pror_{\lambda})$ and $\rer_{\lambda_X} =(\bm{I}-\pror_{\lambda_X})$.

The last two rows of Table \ref{toverviewhd} outline the permutation schemes that we will consider in Sections \ref{secflhd} and \ref{secdres}. The first two rows summarize the Freedman-Lane method discussed in Section \ref{fl} and the Kennedy method \citep{kennedy1995randomization, kennedy1996randomization,winkler2014permutation}.
This table is analogous to Table 2 in \citet{winkler2014permutation} and allows easy comparison of the new methods with the existing methods discussed in  \citet{winkler2014permutation}.

Although Table \ref{toverviewhd} outlines the permutation schemes that we will use, several crucial specifics remain to be filled in.  For example, several choices of the regularization parameters $\lambda$ and $\lambda_X$  can be considered. Moreover, the computational challenge of performing nuisance estimation in every step needs to be addressed. Finally and importantly, we must determine what test statistics are suitable to  use within our permutation tests.

\begin{table}[h!] 
  \begin{center}
    \caption{Permutation schemes for four different methods. The last two methods are novel and can account for high-dimensional nuisance.}
      \label{toverviewhd}
    \begin{tabular}{ll}        
      \hline
      \textbf{Method} & \qquad\textbf{Model after permutation} \\
      \hline
      Freedman-Lane & \qquad $(\Pe\re+\pro)\Y=\X\bm{\beta}+\Z\bm{\gamma}+\bm{\epsilon}$\\
      Kennedy & \qquad  $\Pe\re\Y=\re\X\bm{\beta}+\bm{\epsilon}$\\
      Freedman-Lane HD & \qquad $(\Pe\rer_{\lambda}+\pror_{\lambda})\Y=\X\bm{\beta}+\Z\bm{\gamma}+\bm{\epsilon}$\\
      Double Residualization & \qquad  $(\Pe\rer_{\lambda}+\pror_{\lambda}) \Y=\rer_{\lambda_X}\X\bm{\beta}+\bm{\epsilon}$\\
      \hline
    \end{tabular}
  \end{center}
\end{table}

\subsection{Freedman-Lane HD}  \label{secflhd}
As discussed in Section \ref{fl}, the  low-dimensional Freedman-Lane method is known to provide excellent  type I error control and power. Here we will provide an extension to the case of high-dimensional nuisance. We will refer to this test as \emph{Freedman-Lane HD}. The permutation scheme that we  use is analogous to that of Freedman-Lane and  is shown in the third row of Table \ref{toverviewhd}. 

As in the Freedman-Lane method, after every permutation, we will require nuisance estimation to compute $T_j$.
We will choose ridge regression to do this.
 Note however that when many permutations are used, performing a ridge regression after every permutation can be a large computational burden. We will therefore compute   $\lambda$ only once, for the unpermuted model. 
We take $\lambda$ to be the value that gives the minimal mean cross-validated error; see Section \ref{secsimset} for more details.
 After each permutation, we then use the same parameter $\lambda$ in the ridge regression. Thus, after the $j$-th permutation, to compute the new ridge residuals, we will only need to pre-multiply the transformed outcome $(\Pe_j\rer_{\lambda}+\pror_{\lambda})\Y$ by  $\rer_{\lambda}$. We only need to compute $\rer_{\lambda}$ once.
Owing to this approach, essentially we  need to perform ridge regression only once.

An important consideration is the test statistic $T$ used within the permutation test. The usual \emph{F}-statistic and Wald statistic are only defined when the nuisance is low-dimensional. Extending these definitions to the high-dimensional setting with $q\geq n$  is problematic. For example, a Wald-type statistic would require an unbiased  estimate of $\beta$ and a variance estimate. 
The partial correlation \eqref{parcor}, however, is more naturally generalized to the $q\geq n$ setting: we can replace the residuals $\re\Y$ and   $\re\X$ by the ridge residuals 
$\rer_{\lambda}\Y$ and   $\rer_{\lambda_X}\X$. Similarly we can generalize the semi-partial correlation \eqref{T1OLSsemi}, by replacing $\re\Y$ by $\rer_{\lambda}\Y$.
This gives the following test statistics, which generalize the partial correlation \eqref{parcor} and the semi-partial correlation \eqref{T1OLSsemi} respectively:
\begin{equation} \label{parcorHD}
\rho\big( \rer_{\lambda} \Y,\rer_{\lambda_X}\X \big)=    \frac{  (\rer_{\lambda}\Y-\bm{\mu}_1)'   (\rer_{\lambda_X} \X-\bm{\mu}_2 ) }{  \sqrt{\sum_i (\rer_{\lambda}\Y-\bm{\mu}_1)_i^2  \sum_i (\rer_{\lambda_X}\X-\bm{\mu}_2)_i^2 } },
\end{equation}
\begin{equation} \label{semiparcorHD}
\rho\big( \rer_{\lambda} \Y,\X  \big)=    \frac{  (\rer_{\lambda}\Y-\bm{\mu}_{1})'(\X-\bm{\mu}_{x})  }{ \sqrt{ \sum_i (\rer_{\lambda}\Y-\bm{\mu}_{1})_i^2  \sum_i (\X-\bm{\mu}_x)_i^2  }}.
\end{equation}
Here,  $\bm{\mu}_1$, $\bm{\mu}_2$ and $\bm{\mu}_x$ are  $n$-vectors whose entries are  the sample means of $\rer_{\lambda} \Y$,  $\rer_{\lambda_X}\X$ and $\X$ respectively.  \citet{zhu2018significance} also use a type of generalized partial correlation as the test statistic.

In Section \ref{fl} we reasoned that if $\X$ has been centered, \eqref{parcor} and \eqref{T1OLSsemi} are equivalent with respect to the permutation test. This does not apply to  \eqref{parcorHD} and \eqref{semiparcorHD}. In simulations, using the statistic \eqref{semiparcorHD} tended to result in somewhat higher power than using the statistic  \eqref{parcorHD}. In  Section \ref{secsims} we consider both methods.

In case the generalization of the partial correlation is used, the test statistics $T_1,...,T_w$ on which  Freedman-Lane HD is based are
\begin{equation} \label{eq:flT1par}
T_1 = \rho\big( \rer_{\lambda} \Y,\rer_{\lambda_X} \X  \big),
\end{equation}
\begin{equation} \label{eq:flTjpar}
T_j = \rho\big( \rer_{\lambda} \big(\Pe_j\rer_{\lambda}+\pror_{\lambda})\Y,\rer_{\lambda_X}\X  \big) =
\end{equation}
$$\frac{  \big(\rer_{\lambda} (\Pe_j\rer_{\lambda}+\pror_{\lambda})\Y-\bm{\mu}^j\big)'   (\rer_{\lambda_X} \X-\bm{\mu}_2 ) }{  \sqrt{\sum_i \big(\rer_{\lambda} (\Pe_j\rer_{\lambda}+\pror_{\lambda})\Y-\bm{\mu}^j\big)_i^2  \sum_i (\rer_{\lambda_X}\X-\bm{\mu}_2)_i^2 } }, $$
where $2\leq j \leq w$. Here $\bm{\mu}^j$ is an $n$-vector whose entries are the sample mean of   $\rer_{\lambda}(\Pe_j\rer_{\lambda}+\pror_{\lambda})\Y$.
For the version based on the generalization of the semi-partial correlation, the statistics are
\begin{equation} \label{eq:flT1}
T_1 = \rho\big( \rer_{\lambda} \Y,\X  \big),
\end{equation}
\begin{equation} \label{eq:flTj}
T_j = \rho\big( \rer_{\lambda} (\Pe_j\rer_{\lambda}+\pror_{\lambda})\Y,\X  \big).
\end{equation}
As usual, $T_1$ is just  $T_j$ with  $\Pe_j=\bm{I}_n$.
The pseudo-code for the version based on semi-partial correlations is in Algorithm \ref{a:FLHD}. 

If $q<n$, as $\lambda\downarrow 0$, the test converges to the test for $\lambda=0$, which is the classical Freedman-Lane method.
In the wide range of simulation settings considered in Section \ref{secsims}, the Freedman-Lane HD method stayed on the conservative side, in the sense that the size was less than $\alpha$. This may due to the fact that  if $\lambda>0$ and $2\leq j<k\leq w$,
 the correlation between $T_1$ and $T_j$ tended to be larger than the correlation between $T_j$ and $T_k$ in simulations. This may be related to the fact that 
 the correlation between $\Y$ and $\Y^{*j}$ is strictly larger than the correlation between $\Y^{*j}$ and $\Y^{*k}$, where $\Y^{*j} :=  (\Pe_j\rer_{\lambda}+\pror_{\lambda})\Y$. This inequality is proved in the Supplementary Material.  

As discussed, to perform the test, $\lambda$ and hence $\rer_{\lambda}$ need to be computed only once.
Thus, like the low-dimensional Freedman-Lane procedure, the test requires nuisance estimation after every permutation, but this is not a large computational burden.
The method is often computationally feasible even when many millions of permutations are used; see Section \ref{secsims}. It is also worth mentioning that there exist approximate methods for reducing the number of permutations while still allowing for very small, accurate \emph{p}-values \citep{knijnenburg2009fewer,winkler2016faster}.

\begin{algorithm}[h!] 
\caption{Freedman-Lane HD (version based on semi-partial correlations)}
\begin{algorithmic}[1]  \label{a:FLHD}
\STATE   Compute $\pror_{\lambda}=  \Z(\Z'\Z+\lambda \bm{I}_q)^{-1}\Z'$ and the residual forming matrix $\rer_{\lambda}=\bm{I}-\pror_{\lambda}$. Here $\lambda$ is taken to give the minimal mean cross-validated error (see main text).
\STATE Let $T_1=\rho\big( \rer_{\lambda} \Y,\X  \big)$, the sample Pearson correlation of the $\Y$-residuals with $\X$.
\FOR{$2\leq j \leq w$} 
\STATE Let $T_j= \rho\big( \rer_{\lambda} (\Pe_j\rer_{\lambda}+\pror_{\lambda})\Y,\X  \big)$, where the random matrix $\Pe_j$ encodes  random permutation or sign-flipping. 
\ENDFOR
\STATE The two-sided \emph{p}-value $p$ equals  \eqref{formulap2}.
\RETURN $p$
\end{algorithmic}
\end{algorithm}

\subsection{Double residualization} \label{secdres}
Here we propose a test that we refer to as the \emph{Double Residualization} method. The method is somewhat related to the Kennedy procedure \citep{kennedy1995randomization, kennedy1996randomization,winkler2014permutation}, but not analogous.
The  Kennedy method residualizes both $\Y$ and $\X$  and proceeds to permute the $\Y$-residuals. Here we replace the least squares regression by ridge regression. Moreover, unlike Kennedy's permutation scheme,  we keep $\pror_{\lambda}\Y$ in the model; see Table \ref{toverviewhd}.
The test statistic that we use within the permutation test is the sample correlation. Thus, the test is based on the statistics
$$T_1= \rho \big ( \Y ,\rer_{\lambda_X}\X \big ),$$
\begin{equation} \label{TjDR}
T_j=\rho\big ( (\Pe_j\rer_{\lambda}+\pror_{\lambda})\Y ,\rer_{\lambda_X}\X \big ),
\end{equation}
where $2\leq j \leq w$.
The difference between  \eqref{TjDR} and \eqref{eq:flTjpar} is that \eqref{eq:flTjpar} contains an additional $\rer_{\lambda}$.
The pseudo-code for the Double Residualization method is in  Algorithm \ref{a:DR}.
We take  $\lambda$ and $\lambda_X$ to be the  values that give the minimal mean cross-validated error; see Section \ref{secsimset} for more details.
For fixed $q$, as $n\rightarrow\infty$, the Double Residualization method becomes equivalent  to the Kennedy method and the Freedman-Lane method if the penalty is $o_{\mathbb{P}}(n^{1/2})$, as shown in the Supplementary Material. 
The case that $q>n$ is investigated in Section \ref{secsims}.

\begin{algorithm}[h!] 
\caption{Double Residualization}
\begin{algorithmic}[1]  \label{a:DR}
\STATE Compute $\pror_{\lambda}=  \Z(\Z'\Z+\lambda \bm{I}_q)^{-1}\Z'$ and, analogously, $\pror_{\lambda_X}$. Here $\lambda$ and $\lambda_X$ are determined through cross-validation (see main text).
Let  $\rer_{\lambda}=\bm{I}-\pror_{\lambda}$ and $\rer_{\lambda_X}=\bm{I}-\pror_{\lambda_X}$.
\STATE  Let $T_1= \rho \big ( \Y ,\rer_{\lambda_X}\X \big )$, the sample Pearson correlation of $\Y$ and $\rer_{\lambda_X}\X$.
\FOR{$2\leq j \leq w$} 
\STATE Let $T_j=\rho\big ( (\Pe_j\rer_{\lambda}+\pror_{\lambda})\Y ,\rer_{\lambda_X}\X \big )$, where the random matrix $\Pe_j$ encodes  random permutation or sign-flipping. 
\ENDFOR
\STATE The two-sided \emph{p}-value $p$ equals  \eqref{formulap2}.
\RETURN $p$
\end{algorithmic}
\end{algorithm}

\subsection{Multi-dimensional parameter of interest} \label{multidimbeta}
In the above we considered the case that the tested parameter $\beta$ has dimension $d=1$. Our tests can be extended to the case $d>1$ by using Pesarin's Non-Parametric Combination (NPC) approach \citep[][ch. 4]{pesarin2010permutation}. This is a general method for combining permutation tests of different hypotheses into a test for the intersection hypothesis. 
The NPC principle can be applied in a wide range of scenarios.
In simpler settings with no nuisance, NPC has important proven properties, such as asymptotically optimal power. Here, we will explain how NPC  can be applied in our setting. For convenience, we will focus on the application of NPC to our test of Algorithm \ref{a:FLHD}, i.e., Freedman-Lane HD based on the generalized semi-partial correlation. Combining NPC with our other tests can be done similarly, 
but can be computationally much less efficient for large $d$, as will be explained below. 

Suppose $d>1$. We are interested in $H_0: \bm{\beta}=\bm{\beta}_0$, where we assume $\bm{\beta}_0=\bm{0}$ again for notational convenience. For every $1\leq l \leq d$, let $\beta_l$ be the $l$-th entry of $\bm{\beta}$. The hypothesis of interest $H_0$ is the intersection of $H^1,...,H^d$, where $H^l$ is the hypothesis that $\beta_l$ equals $0$.
To test $H_0=H^1\cap...\cap H^d$, we proceed as follows.
As usual, sample random matrices $\bm{P}_1,...,\bm{P}_w$ that encode permutation (or sign-flipping).
For every $1\leq l \leq d$ and $1\leq j \leq w$, define
 $$T_j^l = \rho\big( \rer_{\lambda} (\Pe_j\rer_{\lambda}+\pror_{\lambda})\Y,\X_{\cdot l}  \big),$$
 where $\X_{\cdot l}$ is the $l$-th column of $\X$.
 A key point here is that the same permutation matrix $\Pe_j$ is used to compute each of the statistics $T_j^1,....,T_j^d$. 
 Due to this manner of simultaneous permutation, the dependence structure of $(T_j^1,....,T_j^d)$ mimics that of $(T_1^1,....,T_1^d)$.
 Indeed, if $\bm{\gamma}$ were exactly known so that we could replace $\rer_{\lambda}\Y$ and $\pror_{\lambda}\Y$ by  
 $\bm{\epsilon}$ and  $\Z\bm{\gamma}$, then 
 $(T_j^1,....,T_j^d)$ and $(T_1^1,....,T_1^d)$ would have exactly the same dependence structure under $H_0$.

Consider a function $\Psi:\mathbb{R}^d\rightarrow \mathbb{R}$, which will be used to compute a combination statistic \citep[][ch. 4]{pesarin2010permutation}. For every $1\leq j \leq w$ define $\Psi_j= \Psi(T_j^1,....,T_j^d)$. Note that if  $\rer_{\lambda}\Y$ and $\pror_{\lambda}\Y$ would be the exact errors and expected values, then under $H_0$, $\Psi_1,...,\Psi_w$ would be identically distributed and exchangeable. 
The  \emph{p}-value for testing $H_0$ is now computed as in \eqref{formulap}  but with $T_j$ replaced by the combination statistic $\Psi_j$. 
The pseudo-code for this test is in Algorithm \ref{a:FLHDmulti}. Note that if $d=1$ and $\Psi$ is the identity and a two-sided \emph{p}-value is computed, then this method reduces to the test of Algorithm \ref{a:FLHD}.

\begin{algorithm}[h!] 
\caption{ Extension of the test of Algorithm \ref{a:FLHD} to the case that $d>1$.}
\begin{algorithmic}[1]  \label{a:FLHDmulti}
\STATE Compute $\pror_{\lambda}=  \Z(\Z'\Z+\lambda \bm{I}_q)^{-1}\Z'$ and the residual forming matrix $\rer_{\lambda}=\bm{I}-\pror_{\lambda}$. Here $\lambda$ is taken to give the minimal mean cross-validated error.
\FOR{$1\leq l \leq d$}
\STATE Let $T_1^l =\rho\big( \rer_{\lambda} \Y,\X_{\cdot l}  \big)$, where $\X_{\cdot l}$ is the $l$-th column of $\X$.
\ENDFOR
\FOR{$2\leq j \leq w$}
\STATE Consider a random $n\times n$ matrix $\Pe_j$ encoding random permutation or sign-flipping.
\STATE Compute $\rer_{\lambda} (\Pe_j\rer_{\lambda}+\pror_{\lambda})\Y$.
\FOR{$1\leq l \leq d$} 
\STATE Let $T_j^l= \rho\big( \rer_{\lambda} (\Pe_j\rer_{\lambda}+\pror_{\lambda})\Y,\X_{\cdot l}  \big)$. 
\ENDFOR
\ENDFOR
\FOR{$1\leq j \leq w$}
\STATE Compute $\Psi_j= \Psi(T_j^1,....,T_j^d)$, where $\Psi$ is the combining function.
\ENDFOR
\STATE The \emph{p}-value \emph{p} equals $w^{-1}\big|\{1\leq j \leq w: \Psi_j\geq \Psi_1\}\big|$.
\RETURN $p$
\end{algorithmic}
\end{algorithm}

The function $\Psi$ should be chosen such that high values of $\Psi_1$ indicate evidence against $H_0$. The choice of $\Psi$ influences power. Examples of functions $\Psi$ are 
$\Psi(t_1,...,t_d)=\max(|t_1|,...,|t_d|)$ and $\Psi(t_1,...,t_d)=d^{-1}\sum_{l=1}^d |t_l|$.  The former choice of $\Psi$ if often used when one or few of the coefficients $\beta_1,...,\beta_d$ are expected to be nonzero under the alternative. Otherwise, the latter choice of $\Psi$ is often used. Other examples of combining functions $\Psi$ are in \citet[][ch. 4]{pesarin2010permutation}.

Applying NPC to the other tests of Sections \ref{secflhd} and \ref{secdres} tends to be computationally less efficient than the method of Algorithm  \ref{a:FLHDmulti}.
For example, applying NPC to our Double Residualization method would require ridge-regressing each of the $d$ variables of interest (corresponding to $\beta_1,...,\beta_d$) on the nuisance variables. 

\section{Simulations} \label{secsims}
We used simulations to gain additional insight into the performance of the new tests, as well as existing tests. The simulations were performed with \emph{R} version 3.6.0 on a server with 40 cores and 1TB RAM. In Section \ref{secsimgaus} we consider scenarios where the outcome $Y$ follows a standard Gaussian high-dimensional linear model. In Section \ref{secsimrob} we consider non-standard settings with  non-normality and  heteroscedasticity.  We consider simulated  datasets where the covariates  have equal variances. It is well-known that when the data are not standardized, this can affect the accuracy of the model obtained with ridge regression  \citep[][p.257]{buhlmann2014high}.

\subsection{Simulation settings and tests}  \label{secsimset}

 We considered the model in Section \ref{secnota}, where the variable of interest  was one-dimensional, i.e., $\beta\in\mathbb{R}$. 
 The case $d>1$ is considered in Section \ref{secsimmulti}. 
 In every simulation, the covariates had mean $0$ and variance $1$.
They were sampled  from a multivariate normal distribution with homogenous correlation $\rho'$, unless stated otherwise.  The errors $\bm{\epsilon}$ had variance 1, unless stated otherwise.  The intercept was $\gamma_1=0$, i.e., $Y$ had mean $0$. The tested hypothesis was $H_0:\beta=0$. 
The sample size in the reported simulations was $n=30$, unless stated otherwise. We obtained comparable results for other sample sizes. 
The estimated probabilities in the tables are based on $10^4$ repeated simulations, unless stated otherwise.

In the power simulations we usually took $|\beta|$ to be relatively large compared to most of the nuisance coefficients. The reason is that testing in high-dimensional models is very challenging. For example, in settings with $|\beta|=|\gamma_1|=...=|\gamma_q| > 0$ the power of all the tests considered (including the competitors) usually  barely exceeds the type I error rate.

The penalty $\lambda$ was chosen to give the minimal mean error, based on 10-fold cross validation. The penalty $\lambda_X$ was chosen analogously. 
To compute the penalties, we used the \emph{cv.glmnet()} function in the R package \emph{glmnet}.
We used $[10^{-5} ,10^{5}]$ as the range of candidate values for the penalty. The penalty obtained with \emph{cv.glmnet()} is scaled by a factor $n$, so we multiplied this penalty by $n$ to  obtain $\lambda$. We included an intercept in the ridge regressions, but excluding the intercept gave very similar results.

All tests used were two-sided. The tests corresponding to the columns of the tables in this section are the following. 

``FLH1" is the Freedman-Lane HD test defined in Section \ref{secflhd}, with test statistics $T_1,...,T_w$ based on the generalized partial correlation  as in \eqref{eq:flTjpar}.
``FLH2" is the same, except that  $T_1,..,T_w$  are based on the generalized \emph{semi}-partial correlation  as in \eqref{eq:flTj}.
``DR"  is the Double Residualization method of Section \ref{secdres}. Each of these tests used $w=2 \cdot 10^4$ permutations.

``BM" is a  high-dimensional test based on ridge projections, proposed in \citet{buhlmann2013statistical}.  This test is based on a bias-corrected estimate $|\hat{\beta}_{\text{corr}}|$ of $|\beta|\in \mathbb{R}$ and an asymptotic upper bound of its distribution. 
We used the implementation in the  R package \emph{hdi} \citep{dezeure2015high}.

``ZZ" is a  high-dimensional test based on Lasso projections, proposed in   \citet{zhang2014confidence}. This method constructs a different bias-corrected estimate $\hat{b}$ of $\beta$, which has an asymptotically known  normal distribution under certain assumptions, such as sparsity. 
For this test we also used  the \emph{hdi} package. We could not include this test in the simulations with a very high number of nuisance parameters, since it is computationally very time-consuming when $q$ is large, as also noted in \citet{dezeure2015high}. We expect the test to have good power in these settings.   

``BO" is the bootstrap approach in \citet{dezeure2017high}, which is also implemented in the \emph{hdi} package. We set the number of bootstrap samples per test to 1000 and considered the robust version of the method. We used the shortcut, which avoids repeated tuning of the penalty.  Still, the method was very slow, so that we used $10^3$  instead of $10^4$ repeated simulations of this method per setting. Also, we did not include the test in the simulations with very large $q$.

\subsection{Gaussian, homoscedastic outcome} \label{secsimgaus}

We first consider some settings with  a moderately large number of nuisance coefficients, $q=60$.
We first simulated a setting with $\gamma_2=...=\gamma_{60}=0.05$, i.e, $\bm{\gamma}$ was dense. 
 We took $\rho'=0.5$. The estimated level and power of the tests described above, for different \emph{p}-value cut-offs $\alpha$,  are shown in Table \ref{table:hdaspr05}.
The tests rejected $H_0$ if the \emph{p}-value was smaller than $\alpha$. The level of  a test should be at most $\alpha$. 

Table \ref{table:hdaspr05} shows that the test ZZ by \citet{zhang2014confidence} was rather anti-conservative. Especially for small $\alpha$, its level was many times larger than $\alpha$. This is partly due to the anti-sparsity. Indeed, ZZ only has proven asymptotic  properties under a sparsity assumption.
The bootstrap approach BO of \citet{dezeure2017high} was much less liberal, but still seemed to be somewhat anti-conservative for small $\alpha$.
Of the other tests, Freedman-Lane HD 2 (FLH2) often had the most power. The Double Residualization method had relatively low power when $\alpha$ was small, e.g. $0.001$.

\begin{table}[!h] \normalsize  
\begin{center}
\caption{Dense setting with $\rho'=0.5$, $n=30$, $q=60$. Power is shown for $\beta= 1.5$. } 
    \begin{tabular}{ l l   l  l l l  l l }    
\hline \\[-0.4cm]
 &  & \multicolumn{5}{l}{\qquad \qquad \qquad \qquad  Method} \\ \cline{3-8} 
 &  $\alpha$ &  FLH1 & FLH2  & DR  & BM  & ZZ & BO  \\ \hline 
  &  0.05   \qquad \quad  & .0281 & .0333 & .0219 & .0087 & .0666 & .063   \\ 
  level  &  0.01  \qquad \quad  & .0042 & .0063 & .0021 & .0024 & .0311 & .023   \\ 
  &  0.001    \qquad \quad  & .0003 & .0006 & 0001 & .0005 & .0121  & .009 \\  \hline   
   &  0.05   \qquad \quad  & .9062 & .9273 & .9616 & .8901 & .9934 & .982 \\ 
 power  &  0.01  \qquad \quad  & .8373 & .8819 & .7984 & .7679& .9799 & .939 \\ 
  &  0.001    \qquad \quad   & .6716 & .7996 & .3263 & .5795  & .9441 & .857 \\  \hline  
    \end{tabular}
\label{table:hdaspr05}
\end{center}
\end{table}

We also considered a setting with very high correlation $\rho'=0.9$, see Table \ref{table:hdspr09}. We took $\gamma_2=\gamma_3=1$ and $\gamma_4=....=\gamma_{60}=0$. The first 4 methods provided appropriate type I error control.  
For small cut-offs $\alpha$, the method ZZ by \citet{zhang2014confidence} was relatively powerful, but also seemed to be somewhat anti-conservative. 
This method seems more suitable for settings where $q$ is many times larger than $n$. 
Among our permutation methods, Freedman-Lane HD 2 had the best power, while incurring few type I errors.
The method BM by  \citet{buhlmann2013statistical} was relatively conservative.

We repeated the same simulation scenario, but with $n=15$ instead of $n=30$. The results are in Table \ref{table:hdspr09n15}. The methods ZZ of \citet{zhang2014confidence} and BO of \citet{dezeure2017high} were very anti-conservative for $\alpha=0.01$ and $\alpha=0.001$. Our methods provided appropriate type I error control.

\begin{table}[!h] \normalsize  
\begin{center}
\caption{ Sparse setting with $\rho'=0.9$, $n=30$, $q=60$. Power is shown for $\beta= 1.5$. } 
    \begin{tabular}{ l l   l  l l l  l l }    
\hline \\[-0.4cm]
 &  & \multicolumn{5}{l}{\qquad \qquad \qquad \qquad  Method} \\ \cline{3-8} 
 &  $\alpha$ &  FLH1 & FLH2  & DR  & BM  & ZZ & BO  \\ \hline 
  &  0.05   \qquad \quad & .0302 & .0270 & .0348 & .0106  & .0358  & .051 \\ 
  level  &  0.01  \qquad \quad & .0050 & .0035 & .0044 & .0013 & .0104 & .012 \\ 
  &  0.001    \qquad \quad & .0003 & .0001 & .0001 & .0000 & .0022  & .002 \\  \hline   
   &  0.05   \qquad \quad & .4494 & .5426 & .4804 & .3234  & .6050 & .554  \\ 
 power  &  0.01  \qquad \quad & .2283 & .3379 & .2135 & .1506 & .4154  & .346 \\ 
  &  0.001    \qquad \quad & .0685 & .1195 & .0445 & .0501 & .2296  & .206  \\  \hline  
    \end{tabular}
\label{table:hdspr09}
\end{center}
\end{table}

\begin{table}[!h] \normalsize  
\begin{center}
\caption{Sparse setting with $\rho'=0.9$, $n=15$, $q=60$. Power is shown for $\beta= 3$. } 
    \begin{tabular}{ l l   l  l l l  l l }    
\hline \\[-0.4cm]
 &  & \multicolumn{5}{l}{\qquad \qquad \qquad \qquad  Method} \\ \cline{3-8} 
 &  $\alpha$ &  FLH1 & FLH2  & DR  & BM  & ZZ & BO  \\ \hline 
  &  0.05   \qquad \quad & .0268 & .0244 & .0294 & .0030  & .0392  & .050 \\ 
  level  &  0.01  \qquad \quad & .0048 & .0030 &  .0028 & .0004  & .0124  & .026\\ 
  &  0.001    \qquad \quad & .0008 & .0000 & .0000 & .0002  & .0032  & .020 \\  \hline   
   &  0.05   \qquad \quad & .5020 & .6034  & .5090 & .4038 & .7586 & .692 \\ 
 power  &  0.01  \qquad \quad & .2822 & .4558&  .2094 & .2384 & .6248  & .552 \\  
  &  0.001    \qquad \quad & .0730 & .1982 & .0438  & .1244 & .4614  & .386 \\   \hline  
    \end{tabular}
\label{table:hdspr09n15}
\end{center}
\end{table}

Further, we considered a simulation where there were clusters of correlated covariates. The setting was as before, except that there were three independent clusters of size 20. Each cluster had a multivariate normal distribution with all correlations equal to $0.9$. We took $\gamma_2=...=\gamma_{60}=0.05$. The results are in Table \ref{table:clusters}. As before, the tests ZZ of \citet{zhang2014confidence} and BO of \citet{dezeure2017high} had good power, but were anti-conservative.

\begin{table}[!h] \normalsize  
\begin{center}
\caption{ Dense setting with $n=30$, $q=60$ and three clusters of dependent covariates. Power is shown for $\beta= 1.5$. } 
    \begin{tabular}{ l l   l  l l l  l l }    
\hline \\[-0.4cm]
 &  & \multicolumn{5}{l}{\qquad \qquad \qquad \qquad  Method} \\ \cline{3-8} 
 &  $\alpha$ &  FLH1 & FLH2  & DR  & BM  & ZZ & BO  \\ \hline 
  &  0.05   \qquad \quad & .0356& .0224& .0344 & .0130  & .0520 & .073 \\ 
  level  &  0.01  \qquad \quad & .0059 & .0025& .0048 & .0022  & .0248 & .023  \\ 
  &  0.001    \qquad \quad & .0010 & .0002& .0002 &.0007  & .0087 & .008  \\  \hline   
   &  0.05   \qquad \quad & .4892 & .5706  & .5043 & .4188  & 7382 & .620   \\ 
 power  &  0.01  \qquad \quad & .2672 & .3393& .2226 & .2399  & .6199 & .454  \\ 
  &  0.001    \qquad \quad & .0814 & .1007 & .0382 & .0977  & .4741 & .322   \\  \hline  
    \end{tabular}
\label{table:clusters}
\end{center}
\end{table}

We also performed simulations with a very large number of nuisance variables ($q=1000$). We first took $\gamma_2=\gamma_3=1$, $\gamma_4=...=\gamma_{10}=0.2$, $\gamma_{11}=...=\gamma_{1000}=0.$ See Table \ref{table:q1000rho05} for simulations with $\rho'=0.5$ and Table  \ref{table:q1000rho09} for simulations with $\rho'=0.9$. All permutation methods provided appropriate type I error control. 
Double Residualization (DR) had relatively high power for large cut-offs $\alpha$,  but not for small cut-offs. The method BM by  \citet{buhlmann2013statistical} had relatively good power for $\rho'=0.5$ but low power for $\rho'=0.9$.

We also performed simulations where $\gamma$ was very anti-sparse, e.g. with $\gamma_2=1$, $\gamma_3=...=\gamma_{800}=0.002$ and $\rho'=0.9$. 
We also considered negative coefficients and we varied the magnitude of the coefficients and the errors $\bm{\epsilon}$ and the sample size. We also considered more settings where there were multiple independent clusters of correlated covariates.
Also in  these settings, the type I error rate was controlled.

\begin{table}[!h] \normalsize  
\begin{center}
\caption{ Sparse setting  with a large number ($q=1000$) of nuisance variables. Here $\rho'=0.5$, $n=30$. Power is shown for $\beta= 2$. } 
    \begin{tabular}{ l  l l  l  l l    }    
\hline \\[-0.4cm]
 &  & \multicolumn{4}{l}{\qquad \qquad \quad   Method} \\ \cline{3-6} 
 &  $\alpha$ & FLH1 & FLH2 & DR   & BM    \\ \hline 
  &  0.05   \qquad \quad & .0068 & .0065 & .0145 & .0001    \\ 
  level  &  0.01  \qquad \quad & .0013 & .0011 & .0011 & .0000   \\ 
  &  0.001    \qquad \quad & .0002 & .0001 & .0000 & .0000    \\  \hline   
   &  0.05   \qquad \quad & .5577 & .5469 & .9613 & .7820    \\ 
 power  &  0.01  \qquad \quad & .5060 & .5043 & .8007 & .6510    \\ 
  &  0.001    \qquad \quad & .3752 & .4049 & .3463 & .4851     \\  \hline  
    \end{tabular}
\label{table:q1000rho05}
\end{center}
\end{table}

\begin{table}[!h] \normalsize  
\begin{center}
\caption{Sparse setting  with a large number ($q=1000$) of nuisance variables and high correlation $\rho'=0.9$. Power is shown for $\beta= 2$. } 
    \begin{tabular}{ l  l l  l  l l    }    
\hline \\[-0.4cm]
 &  & \multicolumn{4}{l}{\qquad \qquad \quad   Method} \\ \cline{3-6} 
 &  $\alpha$ & FLH1 & FLH2 & DR   & BM    \\ \hline 
  &  0.05   \qquad \quad & .0236 & .0319 & .0358 & .0006    \\ 
  level  &  0.01  \qquad \quad & .0040 & .0074 & .0057 & .0000   \\ 
  &  0.001    \qquad \quad & .0003 & .0006 & .0001 & .0000     \\  \hline   
   &  0.05   \qquad \quad & .4766 & .5317 & .7127 & .2115   \\ 
 power  &  0.01  \qquad \quad & .3106 & .4254 & .4137 & .1042    \\ 
  &  0.001    \qquad \quad & .1303 & .2500 & .1344 & .0407     \\  \hline  
    \end{tabular}
\label{table:q1000rho09}
\end{center}
\end{table}

\subsection{Violations of the Gaussian model} \label{secsimrob}

Permutation tests can be robust to violations of the standard linear model, such as non-normality and heteroscedasticity. \citep{winkler2014permutation,hemerik2020robust}
The power of parametric methods is often substantially decreased when the residuals have heavy tails. The power of the permutation tests is more robust to such deviations from normality. This is illustrated in Table \ref{table:exp3}. Here, the data distribution was the same as in the setting corresponding to Table \ref{table:hdspr09}, except that the errors $\bm{\epsilon}$ were not standard normally distributed, but had very heavy (cubed exponential) tails, scaled such that the errors had standard deviation 1. Note in Table \ref{table:exp3} that the permutation and bootstrap methods still had roughly the  same power as at Table \ref{table:hdspr09}, while the power of BM and ZZ  was strongly reduced compared to Table \ref{table:hdspr09}.

\begin{table}[!h] \normalsize  
\begin{center}
\caption{Same sparse  setting as at Table \ref{table:hdspr09} but with very heavy-tailed  errors.  } 
    \begin{tabular}{ l l   l  l l l  l l }    
\hline \\[-0.4cm]
 &  & \multicolumn{5}{l}{\qquad \qquad \qquad \qquad  Method} \\ \cline{3-8} 
 &  $\alpha$ &  FLH1 & FLH2  & DR  & BM  & ZZ & BO  \\ \hline 
  &  0.05   \qquad \quad & .0345 & .0313 & .0336 & .0034  & .0215 & .022  \\ 
  level  &  0.01  \qquad \quad & .0059 & .0051 & .0053 & .0001  & .0043  & .004  \\ 
  &  0.001    \qquad \quad & .0005 & .0002 & .0002 & .0000  & .0006   & .002  \\  \hline   
   &  0.05   \qquad \quad & .4498 & .5493 & .4593 & .2173  & .5433  & .566   \\ 
 power  &  0.01  \qquad \quad & .2295 & .3353 & .2016 & .0730 & .3173   & .390   \\ 
  &  0.001    \qquad \quad & .0780 & .1309 & .0492 & .0151   & .1374   & .215  \\  \hline  
    \end{tabular}
\label{table:exp3}
\end{center}
\end{table}

As a second type of violation of the standard linear model, we considered heteroscedasticity. We simulated  errors $\epsilon_i$ which were normally distributed, but with standard deviation proportional to the absolute value covariate of interest, $|X_i|$.  We again took $\gamma_2=\gamma_3=1$, $\gamma_4=...=\gamma_{60}=0$.    We took $\rho'=0$ for illustration, since in that case  the method ZZ by \citet{zhang2014confidence}  turned out to be very anti-conservative under heteroscedasticity. Otherwise, the simulated data were again as those used for  Table \ref{table:hdspr09}. 
The results are in Table  \ref{table:hete}.
Note that despite the heteroscedasticity, the permutation-based tests provided appropriate type I error control.  
The bootstrap approach BO of \citet{dezeure2017high} seemed to be anti-conservative for small $\alpha$.
The test BM from \citet{buhlmann2013statistical} had higher power than the permutation methods in this specific setting, but was anti-conservative for small $\alpha$.

In the simulations underlying Table \ref{table:hete}, we did not use sign-flipping, which is known to be robust to heteroscedasticity \citep{winkler2014permutation,hemerik2020robust}.  Surprisingly,  our tests nevertheless provided appropriate type I control.
We also performed these simulations with sign-flipping instead of permutation (results not shown), which further reduced the level of our tests, but also somewhat reduced the power.

\begin{table}[!h] \normalsize  
\begin{center}
\caption{Sparse setting with heteroscedastic errors, $\rho'=0$, $n=30$, $q=60$. Power is shown for $\beta= 1.5$. } 
     \begin{tabular}{ l l   l  l l l  l l }    
\hline \\[-0.4cm]
 &  & \multicolumn{5}{l}{\qquad \qquad \qquad \qquad  Method} \\ \cline{3-8} 
 &  $\alpha$ &  FLH1 & FLH2  & DR  & BM  & ZZ & BO  \\ \hline 
  &  0.05   \qquad \quad & .0352 & .0354 & .0271 & .0338   & .1490   & .077  \\ 
  level  &  0.01  \qquad \quad & .0065 & .0069 & .0050 & .0109   & .0648  & .028  \\ 
  &  0.001    \qquad \quad & .0010 & .0009 & .0008 & .0029  & .0280    & .011  \\  \hline   
   &  0.05   \qquad \quad & .7901 & .8060 & .7855 & .9403   & .9902   & .982  \\ 
 power  &  0.01  \qquad \quad & .6787 & .6861 & .6454 & .8534 & .9741  & .936     \\ 
  &  0.001    \qquad \quad & .4910 & .4909 & .4498 & .6903   & .9332   & .830   \\  \hline  
    \end{tabular}
\label{table:hete}
\end{center}
\end{table}

\subsection{Multi-dimensional parameter of interest} \label{secsimmulti}

We simulated the test  of Section \ref{multidimbeta} for multi-dimensional $\bm{\beta}$. As the combination statistic we used $\Psi(t_1,...,t_d)=\max(t_1,...,t_d)$. The parameter of interest $\bm{\beta}$ had dimension 10 and there were 490 nuisance variables, i.e., $\dim(\bm{\gamma})=491$, since 
$\gamma_1$ is the intercept. The outcome $Y$ followed a Gaussian model, as in Section \ref{secsimgaus}. We considered three simulation settings.
The nuisance parameters were $\gamma_2=3,\gamma_3=2,\gamma_4=1,\gamma_5=...=\gamma_{491}=0$ in the first two settings and $\gamma_2=....=\gamma_{101}=0.03,\gamma_{102}=...=\gamma_{491}=0$ in the third setting. The covariates had a multinormal distribution with homogeneous correlation $\rho'=0.5$ in the first setting and  $\rho'=0.9$ in the last two settings.
The results are in Table \ref{table:simmulti}. The test provided appropriate type I error control.

We conclude from the simulations of Section \ref{secsims} that our tests provide good type I error control and are rather robust to several types of model misspecification. The method ZZ from \citet{zhang2014confidence} was often relatively powerful, but was quite anti-conservative in several scenarios. 
The bootstrap approach BO of \citet{dezeure2017high} was also anti-conservative in several scenarios, but less so.
The method BM from \citet{buhlmann2013statistical} tended to be relatively conservative.

\begin{table}[!h] \normalsize  
\begin{center}
\caption{Multi-dimensional $\bm{\beta}\in\mathbb{R}^{10}$. Power is shown for $\bm{\beta}=(3,2,1,0,...,0)$. } 
    \begin{tabular}{ l  l l  l  l l    }    
\hline \\[-0.4cm]
 &  & \multicolumn{3}{l}{\qquad  \quad   Simulation setting} \\ \cline{3-5} 
 &  $\alpha$ & Setting 1& Setting 2 & Setting 3     \\ \hline 
  &  0.05   \qquad \quad & .0174 & .0197 & .0330   \\ 
  level  &  0.01  \qquad \quad & .0023 & .0024 & .0055     \\ 
  &  0.001    \qquad \quad & .0004 & .0002 & .0002        \\  \hline   
   &  0.05   \qquad \quad & .4443 & .5098 & .6286      \\ 
 power  &  0.01  \qquad \quad & .3740 & .4552 & .5731       \\ 
  &  0.001    \qquad \quad & .2503 & .3788 & .4736        \\  \hline  
    \end{tabular}
\label{table:simmulti}
\end{center}
\end{table}

\section{Data analysis} \label{secdata}

We analyze a dataset about riboflavin (vitamin B2) production with \emph{B. subtilis}. This dataset is called \emph{riboflavin} and is publicly available \citep{buhlmann2014high}. It contains normalized  measurements of expression rates of 4088 genes from $n=71$ samples. We use these as input variables. Further, for each sample the dataset contains the logarithm of the riboflavin production rate, which is our one-dimensional outcome of interest. We (further)  standardized the expression levels by subtracting the means and dividing by the standard deviations. We also shifted the outcome values to have mean zero.

For every $1\leq i \leq 4088$, we tested the hypothesis $H_i$ that the outcome was independent of the expression level of gene $i$, conditional on the other expression levels. We used the same tests as considered in the simulations. This time we used  $w=2\cdot10^5$ permutations per test.

The results of the analysis are summarized in Table \ref{table:dataan}. The columns correspond to the same methods as considered in Section \ref{secsims}.    For every method, the fraction of rejections is shown for different \emph{p}-value cut-offs $\alpha$. The fraction of rejections is the number of rejected hypotheses divided by 4088, the total number of hypotheses. The hypotheses that were rejected, were those with \emph{p}-values smaller than or equal to the cut-off $\alpha$.

With most methods we obtain many  \emph{p}-values smaller than 0.05. This is not the case for the test BM by \citet{buhlmann2013statistical}, which is known to be relatively conservative.
After Bonferroni's multiple testing correction, we reject no hypotheses with any method, suggesting there is no strong signal in the data. \citet{van2014asymptotically} also obtained such a result with this dataset.

\begin{table}[!ht] \normalsize  
\caption{Real data analysis. For different \emph{p}-value cut-offs $\alpha$, the fraction of rejected hypotheses is shown. } 
\begin{center}
    \begin{tabular}{   l  l  l l l  l l }    
\hline \\[-0.4cm]
   & \multicolumn{5}{l}{\qquad    Fraction of rejected hypotheses} \\ \cline{2-7} 
  $\alpha$  & FLH1 &  FLH2 & DR     & BM & ZZ & BO  \\ \hline 
    0.05    \qquad \quad  & .0005 & .0259 & .0428 &  0  & .0135 & .0272 \\  
     0.01   \qquad \quad  & 0 & .0071 & .0066 &  0  & .0022  & .0051 \\ 
    0.001  \qquad \quad  & 0 & .0002 & .0012 &  0  & .0007  & .0024  \\ 
    0.0001    \qquad \quad  & 0 & 0 & 0 &  0  & 0  & 0 \\  \hline  
    \end{tabular}
\label{table:dataan}
\end{center}
\end{table}

\section{Discussion}
We have proposed novel permutation  methods for testing in linear  models, where the number of nuisance variables may be much larger than the sample size.
Advantages of permutation approaches include robustness to certain violations of the standard linear model and compatibility with powerful permutation-based multiple testing methods.

We have proposed two novel permutation approaches, Freedman-Lane HD and Double Residualization. Within these approaches some variations are possible, with respect to how the regularization parameters are chosen and which test statistics are used. 
Our methods  provided excellent type I error rate control in a wide range of simulation settings. 
In particular we considered settings with anti-sparsity, high correlations among the covariates, clustered covariates, fat-tailedness of the outcome variable and heteroscedasticity.
The simulation study was limited to settings with multivariate normal covariates. Future research may address more scenarios.

We compared our methods to the parametric tests in  \citet{buhlmann2013statistical}  and \citet{zhang2014confidence} and to the bootstrap approach in \citet{dezeure2017high}. 
 One advantage of  our methods compared to those in \citet{buhlmann2013statistical} and \citet{zhang2014confidence}, is that they are defined in the case that the parameter of interest is multi-dimensional. 
Further, our tests tended to have higher power than the method by \citet{buhlmann2013statistical}. 
The test by \citet{zhang2014confidence} had relatively good power, but was rather anti-conservative in several scenarios, for example under anti-sparsity and heteroscedasticity.
The bootstrap approach of \citet{dezeure2017high} was also anti-conservative in some scenarios, but  less so. 
Our permutation tests provided appropriate type I error control in all scenarios. Moreover, our permutation tests were computationally much  faster than the bootstrap method.


\setlength{\bibsep}{3pt plus 0.3ex}  
\def\bibfont{\small}  

\bibliographystyle{biblstyle}
\bibliography{bibliography}

\newpage
\section*{Supplementary material}

We show that for fixed $q$, our Double Residualization method is asymptotically equivalent to the Kennedy method under local alternatives if the penalty is $o_{\mathbb{P}}(n^{1/2})$. That method is defined if $q<n$ and  is based  on the statistics $T_j^{K}=\rho\big( \Pe_j \re\Y,  \re\X \big)$, $1\leq j \leq w$ \citep{anderson2001permutation}. Note that the Kennedy method is also asymptotically equivalent to the Freedman-Lane method \citep{anderson2001permutation}.

\begin{proposition} \label{lambdatozero}
Let $\xi \in \mathbb{R}$ and suppose 
$\beta=\xi n^{-1/2}$.   Assume  $\lambda=\lambda_n=o_{\mathbb{P}}(n^{1/2})$ and $\lambda_X=\lambda_{X,n}=o_{\mathbb{P}}(n^{1/2})$ .
Let $G=G_n$ be the group of $n!$ permutation maps and let the $n\times n$ matrices  $\Pe_1,...,\Pe_w$ encode the random permutations as usual. 
Assume for convenience that $\Z$ contains a column of 1's.
Assume that for $1\leq i \leq n$, $\mathbb{E}|(\re\X)_i|^{3}$ and $\mathbb{E}|(\re\Y)_i|^{3}$ are finite.

Consider the Double Residualization method, which rejects $H_0$ when the p-value (2) satisfies $p\leq\alpha$, i.e.,  when the event
$$ \Big\{  w^{-1}|\{ 1\leq j \leq w: T_j \leq T_1  \}|\leq \alpha/2 \Big\} \cup \Big\{  w^{-1}|\{ 1\leq j \leq w: T_j \geq T_1  \}|\leq \alpha/2 \Big\}$$
occurs. 
This test is asymptotically equivalent with the Kennedy method, i.e., as $n\rightarrow\infty$, the difference of the rejection functions converges to 0 in probability.
In particular, as $n\rightarrow \infty$, the level of our test
converges to $2\lfloor w \alpha /2   \rfloor/w\leq \alpha$, where $2\lfloor w\alpha /2   \rfloor/w=\alpha$ if $\alpha$ is a multiple of $2/w$.
\end{proposition}

\begin{proof}
Suppose that $n>q$ and $H_0$ holds.
Let $\hat{\bm{\gamma}}_r^n$ and $\hat{\bm{\gamma}}^n$ be the ridge and least squares estimates  $ (\Z'\Z+\lambda \bm{I}_q)^{-1}\Z'\Y$ and $(\Z'\Z)^{-1}\Z'\Y$ respectively, the latter of which exists with probability 1.
By equations (2.4) and (2.7) in \citet{hoerl1970ridge}, 
$$ \hat{\bm{\gamma}}_r^n = \big[\bm{I}_q  -\lambda_n (\Z'\Z+\lambda_n \bm{I}_q)^{-1}\big]\hat{\bm{\gamma}}^n,$$
so that
$$\hat{\bm{\gamma}}_r^n - \hat{\bm{\gamma}}^n= -\lambda_n (\Z'\Z+\lambda_n \bm{I}_q)^{-1}\hat{\bm{\gamma}}^n =  o_{\mathbb{P}}(n^{1/2}n^{-1})=o_{\mathbb{P}}(n^{-1/2}).$$

Let  $1\leq j \leq w$ and $T_j^{OLS}=\rho\big(( \Pe_j \re + \pro)     \Y,\re\X\big)$. This equals $T_j$  if $\lambda=0.$
As $n\rightarrow\infty$, the product of the sample standard deviations of    $ \Pe_j \rer_{\lambda}\Y   + \pror_{\lambda}\Y      $  and $\rer_{\lambda_X}\X$  converges to a constant $c$, say.
Thus
\begin{align*}
\sqrt{n}T_j=&\sqrt{n}n^{-1}( \Pe_j \rer_{\lambda}\Y+\pror_{\lambda}\Y     -\bm{\mu}_y  )'(  \rer_{\lambda_X}\X-\bm{\mu}_2 )/c + o_{\mathbb{P}}(1),\\
\sqrt{n}T_j^{OLS}=&\sqrt{n}n^{-1}( \Pe_j \re\Y + \pro\Y -\bm{\mu}_y   )'  \re\X/c + o_{\mathbb{P}}(1),
\end{align*}
where $\bm{\mu}_y$ and $\bm{\mu}_2$ denote the $n$-vectors  with entries equal to the sample means of $ \Y$ and $\rer_{\lambda_X}\X$  respectively.

Note that the entries of 
$$ (  \Pe_j \rer_{\lambda}\Y + \pror_{\lambda}\Y  -\bm{\mu}_y ) - ( \Pe_j \re\Y  +  \pro\Y   -\bm{\mu}_y  ) = -\Pe_j\Z(\hat{\bm{\gamma}}_r^n-\hat{\bm{\gamma}}^n) + \Z(\hat{\bm{\gamma}}_r^n- \hat{\bm{\gamma}}^n) 
$$
are $o_{\mathbb{P}}(n^{-1/2})$ and likewise the entries of  $( \rer_{\lambda_X}\X-\bm{\mu}_2 )- \re \X $.
It follows that 
$$\sqrt{n}T_j-\sqrt{n}T_j^{OLS} = \sqrt{n} n^{-1} o_{\mathbb{P}}(n n^{-1/2})=o_{\mathbb{P}}(1).$$

The product of the sample standard deviations of    $ \Pe_j \re\Y $  and $\re\X$  converges to a constant $c'$, say.
Note that

$$  \sqrt{n}T_j^{OLS}c = \sqrt{n} T_j^{K} c' + \sqrt{n}n^{-1}(\pro\Y  -\bm{\mu}_y)' \re\X +o_{\mathbb{P}}(1)=\sqrt{n} T_j^{K} c' + o_{\mathbb{P}}(1),$$

 since  $(\pro\Y)' \re\X =0$ and $\bm{\mu}_y' \re\X =0$.

Hence the two tests are asymptotically equivalent.

Under $\xi=0$, the vector $(\sqrt{n}T_1^{K},...,\sqrt{n}T_w^{K})$ is known to have an asymptotic $N(\bm{0},\bm{I}_w)$ distribution \citep{anderson2001permutation}. It follows that $\sqrt{n}T_1,...,\sqrt{n}T_w$ are asymptotically normal and i.i.d..
By the basic Monte Carlo testing principle, if continuous statistics $T_1',...,T_w'$ are i.i.d. under the null hypothesis, then plugging these statistics into the  \emph{p}-value  formulas in Section 2.1 gives \emph{p}-values which are exact. In case the one-sided \emph{p}-value  is used, this means that  $\mathbb{P}(p\leq c)= c$ when $c\in(0,1)$ is a multiple of $w^{-1}$.
In case the two-sided \emph{p}-value  is used,  then $\mathbb{P}(p\leq c)= c$ when $c\in(0,1)$ is a multiple of $2w^{-1}$.
With the continuous mapping theorem \citep{van1998asymptotic} it follows that plugging  $T_1$,...,$T_w$ into the \emph{p}-value  formulas in in Section 2.1 gives \emph{p}-values 
which are asymptotically exact.  
Thus the probabilities
$\mathbb{P}\big(  w^{-1}|\{ j: T_j \leq T_1  \}|\leq \alpha/2 \big)$ and $\mathbb{P}\big( w^{-1}|\{ j: T_j \geq T_1  \}|\leq \alpha/2 \big)$ both converge to
$\lfloor w\alpha /2   \rfloor/w$.
\end{proof}

\bigskip\bigskip\bigskip
In Section 3.1, we refer to the proposition below.
Let $2\leq j <k \leq w$. 
We will write e.g. $cor(\Y,\Y^{*j})$ for the \emph{true} correlation of the entries of $\Y$ and $\Y^{*j}$, i.e., the true correlation of $\Y_i$ and $\Y^{*j}_i$, which is the same for every $1\leq i \leq n$.
Similarly we denote true covariances and variances using $cov$ and $var$.

\begin{proposition} \label{corlarger}
Let $\lambda>0$ freely depend on the data. 
Assume the entries of   $\pror_{\lambda} \Y$ have expected value $0$. 
 Let $2\leq j <k \leq w$. Then $cor(\Y, \Y^{*j})> cor(\Y^{*j}, \Y^{*k})$. 
\end{proposition}

\begin{proof}
Let $\U\D\V'$ be the singular value decomposition of $\Z$. Here $\D$ is an $n\times q$ pseudo-diagonal matrix.  Its diagonal entries are nonzero, since $\Z$ has full rank (with probability 1).
Then $\pror_{\lambda}$ equals
\begin{align*}
\Z\big(&\Z'\Z+\lambda\big)^{-1}\Z'=\\
 \U\D\V'\big(&\V\D'\U'\U\D\V'+\lambda\big)^{-1}\V\D'\U'=\\
 \U\D\V'\big(&\V(\D'\D+\lambda)\V'\big)^{-1}\V\D'\U'.
\end{align*}
Using $\B^{-1}\A^{-1}=(\A\B)^{-1}$  twice  shows that the above equals
\begin{align*}
 \U\D\V'\V\big(&\D'\D+\lambda\big)^{-1}\V'\V\D'\U'   =  \\
\U\D\big(&\D'\D+\lambda\big)^{-1}\D'\U'.  
\end{align*}
Hence the diagonal matrix $\D(\D'\D+\lambda)^{-1}\D$ contains the singular values of $\pror_{\lambda}$, i.e., the eigenvalues. Note that these lie in $(0,1)$.
  
Thus  $\pror_{\lambda}^2$, which has the same eigenvectors as $\pror_{\lambda}$, has strictly smaller sorted eigenvalues. Consequently $\pror_{\lambda} - \pror_{\lambda}^2$ is positive definite.
Since the entries of $\Y$ and $\pror_{\lambda}\Y$ have expected value 0, so do the entries of $\rer_{\lambda}\Y$.
We have
$$cov(\rer_{\lambda}\Y,\pror_{\lambda}\Y) =       \mathbb{E}n^{-1} (\rer_{\lambda}\Y)'  \pror_{\lambda} \Y=$$
\begin{equation}  \label{covRYHYg0}
      \mathbb{E} n^{-1}\Y' \rer_{\lambda}    \pror_{\lambda} \Y =  \mathbb{E} n^{-1} \Y' (  \pror_{\lambda} - \pror_{\lambda}^2 )\Y > 0,
\end{equation} 
since $\pror_{\lambda} - \pror_{\lambda}^2$ is positive definite. 
We then also have
\begin{equation} \label{covYHY}
cov(\Y, \pror_{\lambda}\Y)= cov(\rer_{\lambda}\Y,\pror_{\lambda}\Y)+ cov(\pror_{\lambda}\Y,\pror_{\lambda}\Y)> 0.
\end{equation}

Note that 
\begin{equation} \label{covYYj}
cov(\Y, \Y^{*j})=   cov(\Y, \pror_{\lambda}\Y)  + cov(\Y, \Pe_j\rer_{\lambda}\Y)=cov(\Y, \pror_{\lambda}\Y),
\end{equation}
 since $\Pe_j\rer_{\lambda}\Y$ is a random permutation of $\rer_{\lambda}\Y$.
Similarly we have 
$$var(\Y^{*j})  =  var(\pror_{\lambda}\Y) +2 cov(\pror_{\lambda}\Y, \Pe_j\rer_{\lambda}\Y)  +  var( \Pe_j \rer_{\lambda}\Y)=$$
\begin{equation} \label{varYj}
 var(\pror_{\lambda}\Y) + 0 +  var( \rer_{\lambda}\Y)
\end{equation}
and 
\begin{equation} \label{covYkYj}
cov(\Y^{*k},\Y^{*j}) = var(\pror_{\lambda}\Y).
\end{equation}
By \eqref{covYYj} and \eqref{varYj},
\begin{equation} \label{corYYj}
cor(\Y, \Y^{*j})= \frac{ cov(\Y,\Y^{*j})  }{  \sqrt{var(\Y)var(\Y^{*j})}    } = \frac{  cov(\Y, \pror_{\lambda}\Y) }{\sqrt{ var(\Y) \big(  var(\pror_{\lambda}\Y)+ var( \rer_{\lambda}\Y  )    \big)             }}.
\end{equation}
By  \eqref{varYj} and \eqref{covYkYj},
$$  cor(\Y^{*k}, \Y^{*j}) = \frac{ cov(\Y^{*k},\Y^{*j})  }{  \sqrt{var(\Y^{*k})var(\Y^{*j})}    }  = \frac{ var(\pror_{\lambda}\Y) }{ var(\pror_{\lambda}\Y)+ var(\rer_{\lambda}\Y)   },   $$
so that $cor(\Y^{*k}, \Y^{*j}) =C  \cdot cor(\Y, \Y^{*j}) $, where
$$ C=   \frac{  var(\pror_{\lambda}\Y)   \sqrt{ var(\Y)}    }{  cov(\Y, \pror_{\lambda}\Y)  \sqrt{  var(\pror_{\lambda}\Y)+ var( \rer_{\lambda}\Y  )                }   } .   $$
Here, $cor(\Y, \Y^{*j})> 0$ by \eqref{covYHY} and \eqref{corYYj}.

We are done if we show that $C<1$. let
\begin{align*}
a&=  var(\pror_{\lambda}\Y)>0,\\ 
b&= var(\rer_{\lambda}\Y)>0   ,\\
c&= cov(\rer_{\lambda}\Y,\pror_{\lambda}\Y ) > 0,\\
\end{align*}
where $c>0$ due to \eqref{covRYHYg0}. Note that 
$$var(\Y)= var(\pror_{\lambda}\Y+\rer_{\lambda}\Y) =a+b+2c,$$
$$cov(\Y, \pror_{\lambda}\Y)= var(\pror_{\lambda}\Y) +cov(\rer_{\lambda}\Y,\pror_{\lambda}\Y )=a+c, $$
so that  $$C=\frac{a\sqrt{a+b+2c}}{(a+c)\sqrt{a+b}}.$$ Fix $a>0$ and $b>0$. For $c\geq 0$, write  $f_1(c)=a\sqrt{a+b+2c}$ and $f_2(c) = (a+c)\sqrt{a+b}$. Note that $f_1(0)= f_2(0)$ and   $$f_1'(c)= a(a+b+2c)^{-1/2}<\sqrt{a}< \sqrt{a+b}=f_2'(c).$$
Thus, for $c>0$,  $$C= \frac{f_1(c)}{f_2(c)}=\frac{f_1(0)+ \int_{0}^{c} f_1'(\zeta)d\zeta }{f_2(0)+ \int_{0}^{c} f_2'(\zeta)d\zeta     }<1. $$
\end{proof}

Note that if  we have  $n>q$ and $\lambda=0$, then $cor(\Y, \Y^{*j})= cor(\Y^{*j}, \Y^{*k})$. Indeed, then $c=0$ in the above proof, so that $C=1$.

\end{document}